%% file: Template.tex
\documentclass{article}
\usepackage{spconf,amsmath,graphicx,hyperref}
\usepackage{cite}
\usepackage{amssymb,amsfonts}
\usepackage{algorithmic}
\usepackage{graphicx}
\usepackage{textcomp}
\usepackage[table]{xcolor}

\usepackage{booktabs}
\usepackage{siunitx}
\usepackage{mathtools}
\usepackage{comment}
\usepackage{multirow}
\usepackage{makecell}
\usepackage{arydshln}   
\usepackage{enumitem}
\usepackage{amsthm}
\usepackage[ruled,linesnumbered,vlined]{algorithm2e}
\usepackage{tabularx}   
\usepackage{array}  
\usepackage{nicematrix}

\usepackage{calc}
\newcommand{\centerinwidthof}[2]{%
  \makebox[\widthof{#1}][c]{#2}%
}

\usepackage{tikz}
\usetikzlibrary{calc}
\usetikzlibrary{patterns,patterns.meta}
\usetikzlibrary{arrows.meta}
\usepackage[outline]{contour}
\usetikzlibrary{shapes.callouts}
\usetikzlibrary{decorations.pathreplacing}
\contourlength{1.2pt}

\usepackage[subrefformat=parens]{subcaption}

\usepackage{xspace}
\makeatletter
\DeclareRobustCommand\onedot{\futurelet\@let@token\@onedot}
\def\@onedot{\ifx\@let@token.\else.\null\fi\xspace}

\makeatother

\renewcommand{\sectionautorefname}{Section}
\renewcommand{\subsectionautorefname}{\sectionautorefname}
\renewcommand{\subsubsectionautorefname}{\sectionautorefname}

\newtheorem{proposition}{Proposition}

\let\orgautoref\autoref

\renewcommand{\autoref}[1]
{%
\def\figureautorefname{Fig.}%
\def\subfigureautorefname{\figureautorefname}%
\def\equationautorefname~##1\null{(##1\null)}%
\def\sectionautorefname{Sec.}%
\def\subsectionautorefname{\sectionautorefname}%
\def\subsubsectionautorefname{\sectionautorefname}%
\orgautoref{#1}%
}

\newcommand{\abs}[1]{\left\lvert#1\right\rvert}

\def\appendixautorefname~#1\null{~#1 \null}

\DeclareFontEncoding{LS1}{}{}
\DeclareFontSubstitution{LS1}{stix}{m}{n}
\DeclareSymbolFont{stixletters}{LS1}{stix}{m}{it}
\DeclareMathAccent{\anticausal}{\mathord}{stixletters}{"91}
\DeclareMathAccent{\causal}{\mathord}{stixletters}{"92}
\DeclareMathAccent{\bidirectional}{\mathord}{stixletters}{"95}
\DeclareMathAccent{\noncausal}{\mathord}{stixletters}{"95}

\makeatletter
\newcommand{\figcaption}[1]{\def\@captype{figure}\caption{#1}}
\newcommand{\tblcaption}[1]{\def\@captype{table}\caption{#1}}
\makeatother

\makeatletter 
\newcommand{\linebreakand}{%
  \end{@IEEEauthorhalign}
  \hfill\mbox{}\par
  \mbox{}\hfill\begin{@IEEEauthorhalign}
}
\makeatother 

\def\BibTeX{{\rm B\kern-.05em{\sc i\kern-.025em b}\kern-.08em
    T\kern-.1667em\lower.7ex\hbox{E}\kern-.125emX}}

\newcolumntype{s}{>{\scriptsize}r}
\newcolumntype{O}{%
  r
  @{\;\scriptsize(}
  >{\scriptsize}l<{\scriptsize)}%
}
\newcolumntype{D}{%
  r
  @{\;\scriptsize(}
  s
  @{\scriptsize\;\!/\;\!}
  >{\scriptsize}r<{\scriptsize)}%
}
\newcolumntype{T}{%
  r
  @{\;\scriptsize(}
  s
  @{\scriptsize\;\!/\;\!}
  s
  @{\scriptsize\;\!/\;\!}
  >{\scriptsize}r<{\scriptsize)}%
}
\colorlet{oodcolor}{black!12}

\setlist[itemize]{
  topsep=2pt,      
  itemsep=1pt,     
  parsep=0pt,      
  partopsep=0pt    
}

\title{Diarization Error Decomposition Under Pause Annotation Ambiguity}
\name{Shota Horiguchi, Marc Delcroix, Naohiro Tawara, Alexis Plaquet}
\address{NTT, Inc., Japan}
\begin{document}
\ninept
\abovedisplayskip=2pt
\belowdisplayskip=\abovedisplayskip

\setlength\textfloatsep{5pt}
\setlength\dbltextfloatsep{8pt}
\setlength\floatsep{7pt}
\setlength\dblfloatsep{8pt}
\captionsetup[figure]{skip=2pt}
\captionsetup[table]{skip=0.3pt}
\subcaptionsetup[figure]{skip=2pt}
%
\maketitle
\begin{abstract}
Speaker diarization evaluation is sensitive to ambiguity in pause annotation, which can inflate diarization error rate (DER) or obscure genuine model errors.
We show that morphological closing, which has been used for pause-tolerant diarization evaluation, discards segment-level distinctions.
Instead, we propose an exact, overlap-aware decomposition of standard DER into a pause-attributable component, consisting of errors compatible with pause filling, and a residual core component that can serve as a proxy for intrinsic diarization errors.
The decomposition leaves DER unchanged, while the pause-attributable and core components vary monotonically with the pause threshold and eventually saturate.
Experiments spanning synthetic transformations, annotation mismatch, cross-domain evaluation, and tight-boundary diarization show that the decomposition reveals error sources not apparent from standard DER.%
\end{abstract}
\vspace{-0.4em}
\section{Introduction}\label{sec:introduction}
\vspace{-0.4em}
Speaker diarization is the task of identifying speaker-specific speech activity.
It supports downstream applications including speech separation~\cite{boeddeker2018front} and automatic speech recognition (ASR)~\cite{park2025sortformer,polok2026dicow}, conversation analysis~\cite{veluri2024beyond}, and data generation~\cite{defossez2024moshi,nakata2026duplexchat}.
A major challenge in diarization is annotation.
Fully manual annotation is costly~\cite{ryant2020third}, yet accurate speech-boundary labeling remains difficult~\cite{kraljevski2015comparison}.
This motivates evaluation methods that are robust to annotation ambiguity.

To mitigate the impact of annotation variability on evaluation, collar tolerance has long been used, whereby regions around reference speech boundaries are excluded from scoring~\cite{nist_rt04s}.
However, recent studies suggest that this variability largely stems from how short pauses are treated: ASR corpora favor semantic boundaries, whereas diarization corpora systematically split segments by silence duration~\cite{watanabe2020chime,horiguchi2025can}.
Considerable diversity has also been reported both within annotations of a session and across annotators for the same session~\cite{laurent2026scaling}.
Such pause-annotation inconsistencies can inflate diarization error rate (DER) or obscure genuine model errors when pause-related discrepancies dominate the score, while their importance varies by application, motivating their separate analysis.
Importantly, these effects cannot be eliminated by collar tolerance~\cite{horiguchi2025can}.

The goal of this study is to provide an evaluation method that separates errors attributable to inconsistent pause annotation.
We first show that any metric completely invariant to pause filling necessarily factors through morphological closing of both inputs, discarding detailed segment structure.
We instead decompose DER into pause-attributable components and a residual core component, thereby quantifying errors unexplained by pause inconsistencies (\autoref{fig:overview}).
Unlike closing-based approaches, our method (i) exactly decomposes standard DER even under overlapping speech without changing its value; (ii) yields components that vary monotonically as longer pauses are considered; and (iii) saturates both at finite pause durations.
Experiments across synthetic and real-world settings demonstrate that the proposed method successfully decomposes DER into pause-attributable and core components, improving the interpretability and explainability of diarization performance.
Code is available at \href{https://github.com/nttcslab-sp/der-decomposition}{https://github.com/nttcslab-sp/der-decomposition}.

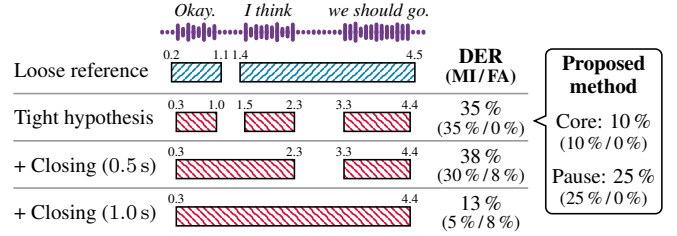
\begin{figure}[t]
\resizebox{\linewidth}{!}{%
\input{figs/overview}%
}
\caption{Example of pause-aware DER evaluation using the closing baseline and proposed method. MI: missed speech; FA: false alarm.}
\label{fig:overview}
\end{figure}

\vspace{-0.4em}
\section{Related work}
\vspace{-0.4em}
DER~\cite{nist_rt03s} is the standard metric for speaker diarization and is inherently duration-weighted, as described in \autoref{sec:der}.
A collar tolerance is often applied around segment boundaries to mitigate annotation uncertainty~\cite{nist_rt04s}.
Jaccard error rate (JER) reduces speaker-level imbalance by averaging speaker-wise Jaccard errors~\cite{ryant2019second}.
Segment-oriented metrics include a segment-based F-measure~\cite{milner2016segment} and conversational DER based on an intersection-over-union (IoU) threshold~\cite{cheng2022conversational}.
Balanced error rate (BER) further combines segment- and duration-level errors with an adaptive IoU threshold~\cite{liu2022ber}.
However, these metrics require additional hyperparameters or more complex scoring procedures, and do not isolate pause-annotation ambiguity.

Prior work addressed pause-annotation ambiguity using forced alignment and morphological closing~\cite{horiguchi2025can}, but forced alignment requires a cumbersome additional evaluation pipeline, and closing has limitations (see \autoref{sec:closing_limitation}).
Separately, word DER reduces sensitivity to temporal annotation ambiguity by evaluating speaker attribution per word~\cite{shafey2019joint}, but requires ASR output and is not a standalone diarization metric.
To our knowledge, our method is the first to exactly decompose standard DER into pause-attributable and residual core components without modifying either input or the resulting DER.

\vspace{-0.4em}
\section{Diarization error rate}\label{sec:der}
\vspace{-0.4em}
Let $\mathcal{R}(t)$ and $\mathcal{H}(t)$ denote the sets of active reference and hypothesis speakers at time $t$, respectively, under the standard one-to-one speaker mapping.
The standard DER components are defined as
\begin{align}
M(t)&=\max(\abs{\mathcal{R}(t)}-\abs{\mathcal{H}(t)},0),\\
F(t)&=\max(\abs{\mathcal{H}(t)}-\abs{\mathcal{R}(t)},0),\\
C(t)&=\min(\abs{\mathcal{R}(t)},\abs{\mathcal{H}(t)})-\abs{\mathcal{R}(t)\cap\mathcal{H}(t)},
\end{align}
where $M(t)$, $F(t)$, and $C(t)$ denote missed-speech, false-alarm, and speaker-confusion error counts, respectively.
DER is then given by
\begin{equation}
E=\frac{1}{L}\int (M(t)+F(t)+C(t))dt,\label{eq:der}
\end{equation}
where $L\coloneqq \int \abs{\mathcal{R}(t)}dt$ is the total duration of speech in reference.
Its breakdown for missed speech, false alarm, and speaker confusion is also given by $E_\mathrm{MI}=\frac{1}{L}\int M(t)dt$, $E_\mathrm{FA}=\frac{1}{L}\int F(t)dt$, and $E_\mathrm{CF}=\frac{1}{L}\int C(t)dt$, respectively.

\vspace{-0.4em}
\section{Proposed method}
\vspace{-0.4em}

\subsection{Information loss under complete pause invariance}\label{sec:closing_limitation}
Let \(x(t),y(t)\in\{0,1\}\) be speaker-activity sequences, and let $c_\tau$ fill all internal pauses no longer than $\tau$, i.e., morphological closing.

\begin{proposition}[Factorization through closing]
\vspace{-3pt}
Suppose that an evaluation score $D_\tau(\cdot,\cdot)$ is invariant to filling any single internal pause of duration at most $\tau$ in either argument. Then
\begin{equation}
D_\tau(x,y)=D_\tau(c_\tau(x),c_\tau(y)).\label{eq:closing-factorization}
\end{equation}
\end{proposition}
\begin{proof}
\vspace{-3pt}
Applying the invariance to all pauses of duration at most $\tau$ in $x$ and $y$ yields their closed representations, proving the result.
\end{proof}
\vspace{-3pt}

Hence, any completely pause-invariant score depends only on the closed inputs and discards distinctions removed by $c_\tau$.
In particular, it cannot distinguish a pause-preserving hypothesis from its pause-filled version, nor annotation tightness from indiscriminate filling.
In this paper, we use closing as a baseline, applying $c_\tau$ to the reference, hypothesis, or both before computing DER.

\vspace{-0.4em}
\subsection{Overlap-aware pause-fill decomposition}\label{sec:decomposition}
\vspace{-0.4em}
Given this limitation, rather than making DER invariant to pause filling, we retain standard DER and identify detection errors that admit an explicit one-sided pause-fill explanation.
We decompose the missed-speech and false-alarm errors into pause-attributable and residual core components, leaving speaker confusion unchanged.

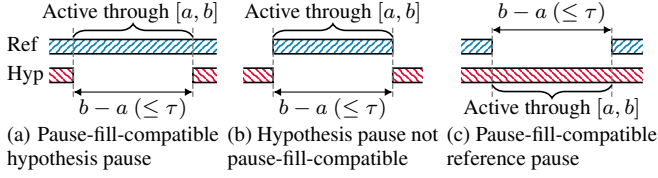
\begin{figure}[t]
\subcaptionbox{Pause-fill-compatible hypothesis pause\label{fig:pause_attributable_mi}}{%
    \input{figs/fill_hyp.tex}%
}\hfill
\subcaptionbox{Hypothesis pause not pause-fill-compatible\label{fig:non_pause_attributable_mi}}{%
    \input{figs/not_fill_hyp.tex}%
}\hfill
\subcaptionbox{Pause-fill-compatible reference pause\label{fig:pause_attributable_fa}}{%
    \input{figs/fill_ref.tex}%
}%
\caption{Pause-fill compatibility with pause-duration threshold $\tau$.}
\end{figure}

For each mapped reference--hypothesis speaker pair, define an interval $g=(a,b)$, with start and end times $a$ and $b$, as an internal \emph{hypothesis} pause if the hypothesis speaker is inactive throughout $g$ but active immediately on both sides.
Let $d\coloneqq b-a$ denote its duration and $\tau\geq 0$ the pause threshold, defined as the maximum duration considered for attribution.
We certify $g$ as pause-fill compatible if i) $d\leq\tau$, ii) the \emph{reference} speaker is active throughout $g$, and iii) a single \emph{reference} activity interval crosses both boundaries of $g$ (\autoref{fig:pause_attributable_mi}).
The third condition excludes isolated reference segments that do not bridge the surrounding hypothesis segments (\autoref{fig:non_pause_attributable_mi}).
The same criterion applies symmetrically to reference pauses (\autoref{fig:pause_attributable_fa}).

Let $P_M(t;\tau)$ denote the number of mapped speakers for which $t$ lies inside a pause-fill-compatible hypothesis pause.
These speakers are active in the reference but inactive in the hypothesis, and thus constitute pause-attributable missed-speech candidates.
Similarly, let $P_F(t;\tau)$ denote the number of mapped speakers for which $t$ lies inside a pause-fill-compatible reference pause.
Under the fixed speaker mapping, the instantaneous total error count is $M(t)+F(t)+C(t)=\max(\abs{\mathcal{R}(t)},\abs{\mathcal{H}(t)})-\abs{\mathcal{R}(t)\cap\mathcal{H}(t)}$.
Filling the $P_M(t;\tau)$ hypothesis activities increases both $\abs{\mathcal{H}(t)}$ and $\abs{\mathcal{R}(t)\cap\mathcal{H}(t)}$ by $P_M(t;\tau)$, reducing the total error by exactly $\min(M(t),P_M(t;\tau))$.
If $P_M(t;\tau)>M(t)$, the remaining candidates only change the error breakdown without reducing the total error, as illustrated in \autoref{fig:pause_fill}.
The same holds symmetrically for $P_F(t;\tau)$.
We therefore define
\begin{align}
    M_\mathrm{pause}(t;\tau)&=\min(M(t),P_M(t;\tau)),\\
    F_\mathrm{pause}(t;\tau)&=\min(F(t),P_F(t;\tau)).
\end{align}

\begin{figure}[t]
\input{figs/pause_attributable.tex}%
\caption{Pause-attributable candidates and the resulting pause-attributable components.}\label{fig:pause_fill}
\end{figure}
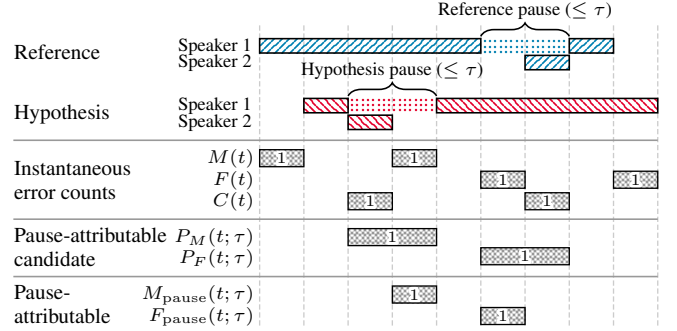

The pause-attributable error $E_\mathrm{pause}$ and the residual core error $E_\mathrm{core}$ under $\tau$ are then defined as 
\begin{align}
E_\mathrm{pause}(\tau)&=\frac{1}{L}\int (M_\mathrm{pause}(t;\tau)+F_\mathrm{pause}(t;\tau))dt,\\
E_\mathrm{core}(\tau)&=\frac{1}{L}\int (M_\mathrm{core}(t;\tau)+F_\mathrm{core}(t;\tau)+C(t))dt.
\end{align}
where $M_\mathrm{core}(t;\tau)=M(t)-M_\mathrm{pause}(t;\tau)$ and $F_\mathrm{core}(t;\tau)=F(t)-F_\mathrm{pause}(t;\tau)$, respectively.
From the definition of each error,
\begin{equation}
E=E_\mathrm{pause}(\tau)+E_\mathrm{core}(\tau)\label{eq:error_sum}
\end{equation}
always holds for an arbitrary threshold $\tau$, and consequently, the proposed method provides a decomposition of DER.
We use $\tau=\infty$ to denote the maximal pause-fill attribution satisfying conditions ii)--iii), with the duration constraint $d\leq\tau$ removed.

\vspace{-0.4em}
\subsection{Theoretical properties}
\vspace{-0.4em}
Closing is sensitive to its width: one-sided closing can yield very large DER, while two-sided closing can locally increase DER or become overly permissive at large widths.
In contrast, the proposed method has the following two properties.

\vspace{-3pt}
\begin{proposition}[Monotonicity with respect to the pause threshold]\label{theorem}
    For any two thresholds $0\leq\tau_1\leq\tau_2$, $E_{\mathrm{pause}}(\tau_1)\leq E_{\mathrm{pause}}(\tau_2)$ and $E_{\mathrm{core}}(\tau_1)\geq E_{\mathrm{core}}(\tau_2)$ hold.
    Hence, with respect to $\tau$, the pause-attributable error is monotonically nondecreasing, whereas the core error is monotonically nonincreasing.
\end{proposition}
\begin{proof}
\vspace{-3pt}
    Since every pause-fill-compatible pause under $\tau_1$ is also pause-fill-compatible under $\tau_2$, $P_M(t;\tau)$ and $P_F(t;\tau)$ are nondecreasing in $\tau$.
    Because $\min(c,x)$ is nondecreasing in $x$ for arbitrary $c$, both $M_{\mathrm{pause}}(t;\tau)$ and $F_{\mathrm{pause}}(t;\tau)$ are nondecreasing in $\tau$.
    Integrating over the scored region gives $E_{\mathrm{pause}}(\tau_1)\leq E_{\mathrm{pause}}(\tau_2)$, and combining this with \autoref{eq:error_sum}, it follows that $E_{\mathrm{core}}(\tau_1)\geq E_{\mathrm{core}}(\tau_2)$.
\end{proof}
\vspace{-3pt}

\begin{proposition}[Finite-threshold saturation]
For any fixed reference and hypothesis with finitely many internal pauses, there exists a finite threshold $\tau_\mathrm{sat}$ such that, for all $\tau\geq\tau_\mathrm{sat}$,
\begin{equation}
E_\mathrm{pause}(\tau)=E_\mathrm{pause}(\infty),\quad
E_\mathrm{core}(\tau)=E_\mathrm{core}(\infty).
\end{equation}
\end{proposition}
\begin{proof}
\vspace{-3pt}
Only condition i), $d\leq\tau$, depends on $\tau$.
Let $\tau_\mathrm{sat}$ be the maximum duration of the internal pauses satisfying the remaining attribution conditions, or $0$ if none exist.
For any $\tau\geq\tau_\mathrm{sat}$, increasing $\tau$ certifies no additional pauses as pause-fill-compatible.
Thus, $P_M(t;\tau)$ and $P_F(t;\tau)$, and consequently $E_\mathrm{pause}(\tau)$ and $E_\mathrm{core}(\tau)$, remain unchanged.
\end{proof}
\vspace{-3pt}

Together, these propositions show that the pause-attributable component increases monotonically with $\tau$, while the core component decreases monotonically, and both eventually saturate.
Thus, unlike closing, increasing $\tau$ changes only how the standard DER is decomposed, without making the evaluation more permissive.

\vspace{-0.4em}
\subsection{Remarks}
\vspace{-0.4em}
\noindent\textbf{Core error is not a stand-alone metric}:
Core error is intended as a diagnostic rather than a replacement for DER.
Optimizing it alone may encourage degenerate predictions that fill within-speaker pauses, reducing core error while increasing pause-attributable error enough to worsen the overall DER.
When the standard DER remains the primary metric, core error can serve as a proxy for intrinsic diarization errors not explained by pause filling.

\noindent\textbf{Dependence on reference tightness}:
The decomposition facilitates system comparison under the same reference, but does not make core error invariant to reference tightness.
Changing the reference can alter the denominator $L$ in \autoref{eq:der}, speaker mapping, and the set of pause-fill-compatible regions.
Core errors computed from references with different annotation tightness are therefore not directly comparable.

\noindent\textbf{Compatibility with collar tolerance}:
The decomposition can be used unchanged with standard collar tolerance, which excludes fixed regions around reference boundaries~\cite{nist_rt04s}.
However, DERs with different collar settings are not directly comparable because both scored errors and the reference speaker-time denominator change.
Alternatively, one could retain the original scoring region and denominator and treat errors within collar regions as a separate collar-attributable component.
We leave this extension for future work to keep the analysis focused on pause-fill attribution and avoid unnecessary complexity.
None of our experiments use collar tolerance.

\noindent\textbf{Extension to other metrics}:
Our method readily extends to other duration-based metrics such as JER~\cite{ryant2019second}.
For each mapped speaker pair, the same pause-fill attribution criterion is applied to separate pause-attributable and core errors, which are then averaged across speakers, yielding an exact decomposition of JER.

\vspace{-0.4em}
\section{Experiments}
\vspace{-0.4em}
\subsection{Controlled synthetic validation}
\vspace{-0.4em}

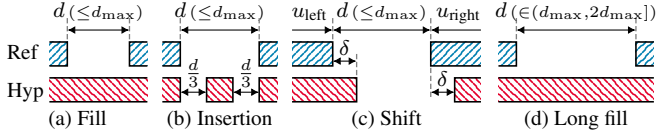
\begin{figure}[t]
\subcaptionbox{Fill\label{fig:fill}}{%
\input{figs/synthetic/complete.tex}%
}\hfill
\subcaptionbox{Insertion}{%
\input{figs/synthetic/partial.tex}%
}\hfill
\subcaptionbox{Shift}{%
\input{figs/synthetic/shift.tex}%
}\hfill
\subcaptionbox{Long fill\label{fig:long_fill}}{%
\input{figs/synthetic/long.tex}%
}%
\caption{Synthetic transformations}\label{fig:synthetic}
\end{figure}

To verify the intended behavior of the proposed method, we begin by computing DERs between the original speech activity labels of the AMI corpus~\cite{carletta2007unleashing} and their synthetically transformed versions.
We apply each of the following transformations speaker-wise (\autoref{fig:synthetic}).
\begin{itemize}[leftmargin=*]
\item \textbf{Fill}: Fill every pause of duration $d\leq d_\mathrm{max}$.
\item \textbf{Insertion}: Activate the middle third of each pause with $d\leq d_\mathrm{max}$.
\item \textbf{Shift}: Shift each pause with $d\leq d_\mathrm{max}$ to the right by $\delta=\frac{1}{2}\min(d,u_\text{left},u_\text{right})$, where $u_\text{left}$ and $u_\text{right}$ are the durations of the adjacent utterances.
\item \textbf{Long fill}: Fill all pauses of duration $d\in(d_\mathrm{max},2d_\mathrm{max}]$.
\end{itemize}
We set $d_\mathrm{max}$, the closing width of the baseline, and $\tau$ to \SI{1}{\second}.
By construction, Fill represents pause-attributable errors; Insertion creates isolated activity rather than bridging the pause, thus representing core errors.
Shift and long fill likewise represent core errors.

\begin{table}[t]
\caption{DER breakdown (\si{\percent}) with synthetic transformations. Reference: original labels; hypothesis: transformed labels. MI: missed speech, FA: false alarm, CF: speaker confusion.}
\label{tbl:synthetic}
\setlength{\tabcolsep}{2pt}
\resizebox{\linewidth}{!}{%
\begin{tabular}{@{}lccccc@{}}
\toprule
&Standard DER&\multicolumn{2}{c}{DER with closing}&\multicolumn{2}{c@{}}{Proposed method}\\\cmidrule(l@{\tabcolsep}r@{\tabcolsep}){2-2}\cmidrule(l@{\tabcolsep}r@{\tabcolsep}){3-4}\cmidrule(l@{\tabcolsep}){5-6}
&&Close ref only&Close ref \& hyp&Core&Pause\\
Transform&\centerinwidthof{0.0}{MI}\,/\,\centerinwidthof{0.0}{FA}\,/\,\centerinwidthof{0.0}{CF}&\centerinwidthof{0.0}{MI}\,/\,\centerinwidthof{0.0}{FA}\,/\,\centerinwidthof{0.0}{CF}&\centerinwidthof{0.0}{MI}\,/\,\centerinwidthof{0.0}{FA}\,/\,\centerinwidthof{0.0}{CF}&\centerinwidthof{0.0}{MI}\,/\,\centerinwidthof{0.0}{FA}\,/\,\centerinwidthof{0.0}{CF}&\centerinwidthof{0.0}{MI}\,/\,\centerinwidthof{0.0}{FA}\\\midrule
(a) Fill&0.0\,/\,0.5\,/\,0.0&0.0\,/\,0.0\,/\,0.0&0.0\,/\,0.0\,/\,0.0&0.0\,/\,0.0\,/\,0.0&0.0\,/\,0.5\\
(b) Insertion&0.0\,/\,0.2\,/\,0.0&0.4\,/\,0.0\,/\,0.0&0.0\,/\,0.0\,/\,0.0&0.0\,/\,0.2\,/\,0.0&0.0\,/\,0.0\\
(c) Shift&0.2\,/\,0.2\,/\,0.0&0.5\,/\,0.0\,/\,0.0&0.0\,/\,0.0\,/\,0.0&0.2\,/\,0.2\,/\,0.0&0.0\,/\,0.0\\
(d) Long fill&0.0\,/\,8.5\,/\,0.0&0.5\,/\,8.4\,/\,0.0&0.0\,/\,8.4\,/\,0.0&0.0\,/\,8.5\,/\,0.0&0.0\,/\,0.0\\
\bottomrule
\end{tabular}%
}
\end{table}

\autoref{tbl:synthetic} shows the results.
Closing changes both DER and its breakdown, and closing both the reference and hypothesis can even remove errors caused by insertion or shift.
In contrast, the proposed method attributes only fill errors to the pause-attributable component, while insertion, shift, and long-fill errors remain in the core component.
In all cases, the two components sum exactly to the standard DER.

\vspace{-0.4em}
\subsection{Loose--tight annotation differences}\label{sec:real1}
\vspace{-0.4em}
\begin{figure}[t]
    \vspace{-0.4em}
    \includegraphics[width=\linewidth]{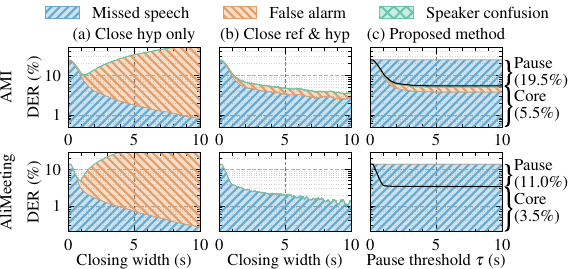}%
    \caption{Pause-aware DERs with the closing baselines and the proposed method. Reference: original loose labels; hypothesis: forced-aligned tight labels. The y-axes use logarithmic scales.}
    \label{fig:loose_vs_tight_annotation}
\end{figure}

We next compare real loose and tight annotations on AMI and AliMeeting (AliM)~\cite{yu2022m2met}, using original loose labels as the DER reference and forced-aligned tight labels as the hypothesis~\cite{horiguchi2025can}.
Most discrepancies are expected to arise from pause filling or boundary errors.
By design, speaker confusion rarely occurs.

\autoref{fig:loose_vs_tight_annotation} compares closing-based approaches with the proposed method.
With one-sided closing applied only to the tight annotation (hypothesis; \hyperref[fig:loose_vs_tight_annotation]{Fig.~\ref*{fig:loose_vs_tight_annotation}(a)}), DER is minimized at corpus-specific widths (\SI{1.24}{\second} for AMI and \SI{0.90}{\second} for AliM).
It reflects a trade-off between reduced misses and increased false alarms, making a common width difficult to choose.
Closing both reference and hypothesis (\hyperref[fig:loose_vs_tight_annotation]{Fig.~\ref*{fig:loose_vs_tight_annotation}(b)}) generally reduces DER but not monotonically and eventually yields nearly zero DER (\SI{0.06}{\percent} for AMI and \SI{0.03}{\percent} for AliM), making the width difficult to choose.
In contrast, the proposed core error (\hyperref[fig:loose_vs_tight_annotation]{Fig.~\ref*{fig:loose_vs_tight_annotation}(c)}) decreases monotonically and converges: full saturation occurs at $\tau=\SI{7.66}{\second}$ and $\SI{2.72}{\second}$ for AMI and AliM, while \SI{95}{\percent} saturation is reached at only \SI{2.02}{\second} and \SI{0.91}{\second}.
Thus, most pause-attributable error is captured at much smaller thresholds, making the decomposition relatively insensitive to the exact choice of a large $\tau$.
Because a universal finite threshold is nevertheless difficult to justify, we henceforth use the saturated, parameter-free endpoint $\tau=\infty$.

\vspace{-0.4em}
\subsection{Train--test label tightness mismatch}
\vspace{-0.4em}
To evaluate the proposed method on a real diarization system, we first use a controlled setting in which the same training data are paired with labels of different annotation tightness.
We use EEND-vector clustering~\cite{kinoshita2021integrating} with the DiariZen architecture without structured pruning~\cite{han2026efficient} (WavLM Large + four Conformer encoders), 16-\si{\second} windows with a 1.6-\si{\second} shift, ResNet speaker embeddings~\cite{wang2023wespeaker}, and VBx clustering~\cite{palka2026vbx}.

Following prior work~\cite{horiguchi2025can}, we use compound sets consisting of AMI mixed headset microphones (AMI-Mix) and single distant microphone (AMI-SD), AliM, MSDWild (MSDW)~\cite{liu2022msdwild}, and VoxConverse (VoxC)~\cite{chung2020spot}.
Supervision labels are defined as follows.
\begin{itemize}[leftmargin=*]
\item $\mathcal{D}_\mathrm{mix}$: Mixed annotation tightness, using original loose labels for AMI and AliM, and original tight labels for MSDW and VoxC.
\item $\mathcal{D}_\mathrm{tight}$: Uniformly tight annotations, using forced-aligned labels for AMI and AliM, and original tight labels for MSDW and VoxC.
\end{itemize}

\begin{table}[t]
\centering
\caption{DERs (\%) of systems trained on compound sets with different label tightness. Train--test tightness mismatches are {\setlength{\fboxsep}{1pt}\colorbox{oodcolor}{shaded}}.}
\label{tbl:training_label_tightness}
\setlength{\tabcolsep}{2.5pt}
\resizebox{\linewidth}{!}{%
\begin{NiceTabular}{@{}llTTD@{}}
\CodeBefore
  \rectanglecolor{oodcolor}{4-2}{4-13}
  \rectanglecolor{oodcolor}{6-2}{6-13}
  \rectanglecolor{oodcolor}{8-2}{8-13}
  \rectanglecolor{oodcolor}{10-2}{10-13}
  \rectanglecolor{oodcolor}{12-2}{12-13}
  \rectanglecolor{oodcolor}{14-2}{14-13}
\Body

\toprule
Corpus&Train&\multicolumn{4}{l}{DER \scriptsize{(MI\,/\,FA\,/\,CF)}}&\multicolumn{4}{l}{Core \scriptsize{(MI\,/\,FA\,/\,CF)}}&\multicolumn{3}{l@{}}{Pause \scriptsize{(MI\,/\,FA)}}\\\midrule
\multicolumn{13}{@{}l@{}}{\textbf{Reference: Original (loose) label}}\\
AMI-Mix~\cite{carletta2007unleashing}&$\mathcal{D}_\mathrm{mix}$&11.9&6.6&3.0&2.2&7.5&3.2&2.1&2.2&4.4&3.4&1.0\\
&$\mathcal{D}_\mathrm{tight}$&27.0&24.1&1.1&1.8&7.9&4.9&1.1&1.8&19.1&19.1&0.0\\
AMI-SD~\cite{carletta2007unleashing}&$\mathcal{D}_\mathrm{mix}$&15.1&8.5&3.3&3.3&10.1&4.5&2.3&3.3&4.9&4.0&1.0\\
&$\mathcal{D}_\mathrm{tight}$&28.9&24.9&1.4&2.7&10.2&6.1&1.4&2.7&18.7&18.7&0.0\\
AliM~\cite{yu2022m2met}&$\mathcal{D}_\mathrm{mix}$&14.3&8.3&2.5&3.5&11.6&6.0&2.0&3.5&2.7&2.2&0.5\\
&$\mathcal{D}_\mathrm{tight}$&22.9&18.5&0.9&3.5&12.8&8.5&0.8&3.5&10.0&10.0&0.0\\\midrule
\multicolumn{13}{@{}l@{}}{\textbf{Reference: Forced-aligned (tight) label}}\\
AMI-Mix~\cite{carletta2007unleashing}&$\mathcal{D}_\mathrm{mix}$&31.9&2.6&26.3&3.0&11.3&2.4&6.0&3.0&20.6&0.2&20.4\\
&$\mathcal{D}_\mathrm{tight}$&11.1&4.8&3.7&2.5&9.8&4.0&3.3&2.5&1.3&0.8&0.5\\
AMI-SD~\cite{carletta2007unleashing}&$\mathcal{D}_\mathrm{mix}$&32.6&3.3&25.1&4.2&13.5&3.1&18.9&4.2&19.2&0.2&18.9\\
&$\mathcal{D}_\mathrm{tight}$&13.6&5.9&4.0&3.7&12.2&5.0&3.5&3.7&1.4&0.9&0.5\\
AliM~\cite{yu2022m2met}&$\mathcal{D}_\mathrm{mix}$&23.9&4.5&14.7&4.6&14.6&4.2&5.8&4.6&9.3&0.4&8.9\\
&$\mathcal{D}_\mathrm{tight}$&15.8&7.4&3.7&4.7&14.3&6.5&3.2&4.7&1.5&0.9&0.5\\\midrule
\multicolumn{13}{@{}l@{}}{\textbf{Reference: Original (tight) label}}\\
MSDW~\cite{liu2022msdwild}&$\mathcal{D}_\mathrm{mix}$&16.7&6.3&4.8&5.6&14.5&5.1&3.7&5.6&2.2&1.2&1.1\\
&$\mathcal{D}_\mathrm{tight}$&17.0&7.2&4.1&5.8&14.7&5.6&3.3&5.8&2.3&1.5&0.8\\
VoxC~\cite{chung2020spot}&$\mathcal{D}_\mathrm{mix}$&8.9&2.9&3.8&2.3&6.6&1.4&2.9&2.3&2.3&1.4&0.9\\
&$\mathcal{D}_\mathrm{tight}$&9.4&3.8&3.0&2.5&6.7&1.6&2.5&2.5&2.6&2.1&0.5\\
\bottomrule
\end{NiceTabular}%
}
\end{table}

\autoref{tbl:training_label_tightness} shows that the DiariZen-based system also reproduces corpus-specific annotation tightness and incurs large DER increases under train--test mismatch, confirming the effect reported in \cite{horiguchi2025can}.
Whereas \cite{horiguchi2025can} analyzed this effect mainly through missed-speech and false-alarm changes and DER after closing, our decomposition preserves the standard DER and shows quantitatively that most of the mismatch-induced increase is pause-attributable, providing a more interpretable account of the mismatch.

\vspace{-0.4em}
\subsection{In-domain vs. out-of-domain evaluation}
\vspace{-0.4em}
We further evaluate whether the proposed decomposition helps interpret out-of-domain degradation using three open-source systems: pyannote 2.1~\cite{bredin2023pyannote}, pyannote 3.1~\cite{plaquet2023powerset}, and DiariZen~\cite{han2026efficient}.
Here, ``out-of-domain'' denotes a corpus not represented in the training data of the corresponding system.
Roughly, pyannote 3.1 improves over 2.1 with a powerset objective, while DiariZen achieves higher accuracy with a larger-scale architecture and self-supervised learning.
Since the three models are trained on different corpora, their comparison is not strictly apples-to-apples.
We generally use official test splits; for LibriCSS~\cite{chen2020continuous}, we use artificial mixtures of the clean source utterances (LibriCSS-Mix) and the official setup of single-channel real recordings (LibriCSS-SD).

\begin{table}[t]
\caption{DERs (\si{\percent}) of open-source systems. Reference: original labels. Results on out-of-domain corpora are {\setlength{\fboxsep}{1pt}\colorbox{oodcolor}{shaded}}.}\label{tbl:results_oss}
\renewrobustcmd{\bfseries}{\fontseries{b}\selectfont}
\renewrobustcmd{\boldmath}{}
\newrobustcmd{\B}{\bfseries}
\setlength{\tabcolsep}{1.95pt}
\centering
\resizebox{\linewidth}{!}{%
\begin{NiceTabular}{@{}lrrrrrrrrr}
\CodeBefore
  \rectanglecolor{oodcolor}{3-8}{3-10}
  \rectanglecolor{oodcolor}{4-2}{5-4}
  \rectanglecolor{oodcolor}{8-2}{8-4}
  \rectanglecolor{oodcolor}{10-2}{10-7}
  \rectanglecolor{oodcolor}{11-2}{17-10}
\Body

\toprule
&\multicolumn{3}{c}{pyannote 2.1~\cite{bredin2023pyannote}}&\multicolumn{3}{c}{pyannote 3.1~\cite{plaquet2023powerset}}&\multicolumn{3}{c@{}}{DiariZen~\cite{han2026efficient}}\\\cmidrule(l{\tabcolsep}r{\tabcolsep}){2-4}\cmidrule(l{\tabcolsep}r{\tabcolsep}){5-7}\cmidrule(l{\tabcolsep}){8-10}
Corpus&DER&Core&Pause&DER&Core&Pause&DER&Core&Pause\\\midrule
AMI-Mix~\cite{carletta2007unleashing}&18.9&14.8&4.1&\B 18.8&14.4&4.4&25.4&\B 9.7&15.7\\
AMI-SD~\cite{carletta2007unleashing}&27.1&21.3&5.9&22.4&17.7&4.7&\B 14.0&\B 9.8&4.2\\
AliMeeting~\cite{yu2022m2met}&27.4&23.9&3.5&24.4&22.2&2.1&\B 12.5&\B 10.4&2.1\\
AISHELL-4~\cite{fu2021aishell}&14.1&11.0&3.1&12.2&10.0&2.2&\B 11.7&\B 8.8&2.9\\
DIHARD-3~\cite{ryant2020third}&26.9&19.2&7.7&21.7&18.7&2.9&\B 14.5&\B 12.4&2.1\\
MSDWild~\cite{liu2022msdwild}&33.1&27.0&6.1&24.9&22.0&3.0&\B 15.6&\B 13.3&2.3\\
VoxConverse~\cite{chung2020spot}&11.2&8.7&2.5&11.3&8.8&2.5&\B 9.2&\B 6.5&2.6\\
NOTSOFAR-1~\cite{vinnikov2024notsofar}&35.5&31.8&3.7&32.2&29.7&2.5&\B 17.9&\B 16.1&1.8\\
DiPCo~\cite{van2020dipco}&41.4&36.8&4.6&\B 35.9&32.0&3.9&37.4&\B 28.1&9.2\\
ICSI~\cite{janin2003icsi}&35.0&25.9&9.1&34.3&25.9&8.4&\B 30.6&\B 24.2&6.4\\
AVA-AVD~\cite{xu2022ava}&64.9&54.3&10.6&48.2&45.3&2.9&\B 42.6&\B 39.4&3.3\\
LibriCSS-Mix~\cite{chen2020continuous}&\B 6.5&5.5&1.0&10.2&6.9&3.3&11.7&\B 4.9&6.8\\
LibriCSS-SD~\cite{chen2020continuous}&13.9&10.6&3.3&13.8&8.6&5.2&\B 12.9&\B 4.3&8.7\\
\bottomrule
\end{NiceTabular}%
}
\end{table}

\autoref{tbl:results_oss} shows that pause-attributable error is generally small in-domain but can be large out-of-domain when annotation tightness mismatches model output.
For example, DiariZen's DER is \SI{25.4}{\percent} on AMI-Mix, but its core error is \SI{9.7}{\percent}, lower than those of the pyannote systems.
This suggests an intrinsic performance advantage obscured by standard DER.
Likewise, cases where DiariZen has a worse DER than pyannote, such as DiPCo~\cite{van2020dipco} and LibriCSS-Mix, can be largely accounted for by increased pause-attributable error.

In contrast, some corpora, including DiPCo (dinner party), ICSI (meetings with up to 10 speakers)~\cite{janin2003icsi}, and AVA-AVD (cinematic content)~\cite{xu2022ava}, exhibited high absolute core error, suggesting substantial domain mismatch with the compound training set.
This distinction could support active learning~\cite{shamsi2023towards}: domains with high absolute core error may be more promising targets for annotation than those dominated by pause-attributable errors.

\vspace{-0.5em}
\subsection{Tight-boundary speaker diarization}
\vspace{-0.4em}

\begin{table}[t]
\centering
\caption{Results of tight-boundary speaker diarization. Reference: forced-aligned tight labels. Results for mismatched train--test label-tightness conditions are {\setlength{\fboxsep}{1pt}\colorbox{oodcolor}{shaded}}.}\label{tbl:results_tight_diarization}
\setlength{\tabcolsep}{2pt}
\resizebox{\linewidth}{!}{%
\begin{NiceTabular}{@{}llTTD@{}}
\CodeBefore
  \rectanglecolor{oodcolor}{2-2}{2-13}
  \rectanglecolor{oodcolor}{6-2}{6-13}
  \rectanglecolor{oodcolor}{10-2}{10-13}
\Body
\toprule
Corpus&Train&\multicolumn{4}{l}{DER\,\scriptsize{(MI\;\!/\;\!FA\;\!/\;\!CF)}}&\multicolumn{4}{l}{Core\,\scriptsize{(MI\;\!/\;\!FA\;\!/\;\!CF)}}&\multicolumn{3}{l@{}}{Pause\,\scriptsize{(MI\;\!/\;\!FA)}}\\\midrule
AMI-Mix&Loose (baseline)&34.0&3.1&26.0&4.7&14.7&3.0&6.9&4.9&19.3&0.1&19.1\\
\cite{carletta2007unleashing}&Tight (topline)&14.0&5.9&3.9&4.3&12.5&5.1&3.1&4.3&1.6&0.8&0.8\\
&Pseudo-tight~\cite{horiguchi2026tight}&18.3&7.8&5.8&4.7&15.7&5.9&5.1&4.7&2.6&1.9&0.7\\
&Pseudo-tight~\cite{horiguchi2026over}&16.4&6.7&5.6&4.1&14.0&5.0&4.9&4.1&2.4&1.7&0.7\\\midrule
AMI-SD&Loose (baseline)&36.3&4.0&25.5&6.8&18.1&3.8&7.5&6.8&18.2&0.2&18.1\\
\cite{carletta2007unleashing}&Tight (topline)&17.1&7.7&4.3&5.0&15.3&6.7&3.6&5.0&1.8&1.1&0.7\\
&Pseudo-tight~\cite{horiguchi2026tight}&21.6&9.7&6.0&5.9&18.9&7.6&5.4&5.9&2.7&2.1&0.6\\
&Pseudo-tight~\cite{horiguchi2026over}&19.6&8.6&5.7&5.3&17.1&6.6&5.2&5.3&2.5&2.0&0.5\\\midrule
AliM&Loose (baseline)&26.5&4.5&14.2&7.8&18.1&4.3&5.9&7.8&8.5&0.2&8.2\\
\cite{yu2022m2met}&Tight (topline)&19.8&8.0&4.8&7.1&18.1&7.1&3.9&7.1&1.7&0.9&0.9\\
&Pseudo-tight~\cite{horiguchi2026tight}&22.9&\multicolumn{3}{@{}>{\scriptsize}l<{\scriptsize)}}{11.2\;\!/\;\!4.9\;\!/\;\!6.8}&20.5&9.5&4.2&6.8&2.4&1.7&0.7\\
&Pseudo-tight~\cite{horiguchi2026over}&21.4&6.7&8.4&6.3&18.2&6.0&5.9&6.3&3.2&0.7&2.5\\
\bottomrule
\end{NiceTabular}%
}
\end{table}

Finally, we analyze tight-boundary diarization, which predicts tight boundaries from loose supervision.
As in prior studies~\cite{horiguchi2026tight,horiguchi2026over}, we train a ReDimNet-B2~\cite{yakovlev2024reshape} followed by an LSTM on $\mathcal{D}_\mathrm{mix}$.
The initial pseudo-labeling method substantially improved DER but yielded limited ASR gains~\cite{horiguchi2026tight}, whereas the enhanced version achieved a modest further DER improvement but much larger ASR gains~\cite{horiguchi2026over}.

\autoref{tbl:results_tight_diarization} shows the results, with the topline trained on $\mathcal{D}_\mathrm{tight}$.
The initial method brings DER closer to the topline but increases core error, which may help explain its limited ASR improvement (e.g., word error rate from \SI{24.32}{\percent} to \SI{23.24}{\percent} for single-channel ASR on AMI-SD~\cite{horiguchi2026over}).
In contrast, the enhanced version reduces core error to a level comparable to or below the loose baseline, consistent with its larger ASR gain (a word error rate of \SI{21.47}{\percent}~\cite{horiguchi2026over}).

\vspace{-0.5em}
\section{Conclusion}
\vspace{-0.5em}
We proposed an exact, overlap-aware decomposition of DER into pause-attributable and residual core components.
The decomposition guarantees monotonicity and finite-threshold saturation.
Analyses spanning synthetic transformations and four real-world settings show how it exposes error sources obscured by standard DER.

\bibliographystyle{IEEEbib-abbrev}
\bibliography{mybib}

\end{document}

%% file: figs/overview.tex
\begin{tikzpicture}[semithick,auto,
label/.style={
    draw=none,
    align=center,
    font=\footnotesize,
    inner sep=0,
    outer sep=0
},
]%

\definecolor{spkblue}{HTML}{0080B1}
\definecolor{spkred}{HTML}{E4002B}
\definecolor{spkgreen}{HTML}{06C755}
\definecolor{spkpurple}{HTML}{6E3B89}

\newcommand{\drawwaveform}[3]{%
  \foreach \i in {#1,...,#2} {
    \pgfmathsetmacro{\x}{0.07*\i}

    \pgfmathsetmacro{\envelope}{
      0.8
      + 0.1*abs(sin(2.3*\i))
      + 0.1*abs(sin(0.7*\i))
    }

    \pgfmathsetmacro{\h}{
      #3*\envelope*(0.1 + 0.15*rnd)
    }

    \draw[
      line width=1.5pt,
      line cap=round,
      color=spkpurple
    ]
      (\x,{-\h}) -- (\x,{\h});
  }%
}

\node[font=\itshape\scriptsize\selectfont,text=black,inner sep=0pt,anchor=base] at ($(1.4em,0.7em)$) {Okay.};
\node[font=\itshape\scriptsize\selectfont,text=black,inner sep=0pt,anchor=base] at ($(4.5em,0.7em)$) {I think};
\node[font=\itshape\scriptsize\selectfont,text=black,inner sep=0pt,anchor=base] at ($(9.2em,0.7em)$) {we should go.};

\pgfmathsetseed{2026}
\drawwaveform{0}{50}{0.05}
\drawwaveform{3}{10}{0.5}
\drawwaveform{16}{25}{0.5}
\drawwaveform{35}{47}{0.5}

\draw[black!40](-6.3em,-2.6em) -- (15.5em,-2.6em);
\draw[black!40](-6.3em,-4.6em) -- (15.5em,-4.6em);
\draw[black!40](-6.3em,-6.6em) -- (15.5em,-6.6em);

\path let \p1 = ($(0.4em,-2.1em)$),
          \p2 = ($(2.5em,-1.3em)$) in
    [draw=black,pattern={Lines[angle=45,distance=2pt]},pattern color=spkblue] (\p1) rectangle (\p2);
\path let \p1 = ($(3.3em,-2.1em)$),
          \p2 = ($(10.7em,-1.3em)$) in
    [draw=black,pattern={Lines[angle=45,distance=2pt]},pattern color=spkblue] (\p1) rectangle (\p2);

\node[font=\footnotesize,text=black,inner sep=0pt,anchor=base west] at ($(-6.3em,-1.9em)$) {Loose reference};
\node[font=\tiny\selectfont,text=black,inner sep=0pt,anchor=base,align=left] at ($(0.4em,-1.15em)$) {0.2};
\node[font=\tiny\selectfont,text=black,inner sep=0pt,anchor=base,align=left] at ($(2.5em,-1.15em)$) {1.1};
\node[font=\tiny\selectfont,text=black,inner sep=0pt,anchor=base,align=left] at ($(3.3em,-1.15em)$) {1.4};
\node[font=\tiny\selectfont,text=black,inner sep=0pt,anchor=base,align=left] at ($(10.7em,-1.15em)$) {4.5};

\path let \p1 = ($(0.6em,-4.2em)$),
          \p2 = ($(2.3em,-3.4em)$) in
    [draw=black,pattern={Lines[angle=-45,distance=2pt]},pattern color=spkred] (\p1) rectangle (\p2);
\path let \p1 = ($(3.5em,-4.2em)$),
          \p2 = ($(5.6em,-3.4em)$) in
    [draw=black,pattern={Lines[angle=-45,distance=2pt]},pattern color=spkred] (\p1) rectangle (\p2);
\path let \p1 = ($(7.7em,-4.2em)$),
          \p2 = ($(10.5em,-3.4em)$) in
    [draw=black,pattern={Lines[angle=-45,distance=2pt]},pattern color=spkred] (\p1) rectangle (\p2);
\node[font=\footnotesize,text=black,inner sep=0pt,anchor=base west] at ($(-6.3em,-3.9em)$) {Tight hypothesis};
\node[font=\tiny\selectfont,text=black,inner sep=0pt,anchor=base,align=left] at ($(0.6em,-3.25em)$) {0.3};
\node[font=\tiny\selectfont,text=black,inner sep=0pt,anchor=base,align=left] at ($(2.3em,-3.25em)$) {1.0};
\node[font=\tiny\selectfont,text=black,inner sep=0pt,anchor=base,align=left] at ($(3.5em,-3.25em)$) {1.5};
\node[font=\tiny\selectfont,text=black,inner sep=0pt,anchor=base,align=left] at ($(5.6em,-3.25em)$) {2.3};
\node[font=\tiny\selectfont,text=black,inner sep=0pt,anchor=base,align=left] at ($(7.7em,-3.25em)$) {3.3};
\node[font=\tiny\selectfont,text=black,inner sep=0pt,anchor=base,align=left] at ($(10.5em,-3.25em)$) {4.4};

\path let \p1 = ($(0.6em,-6.2em)$),
          \p2 = ($(5.6em,-5.4em)$) in
    [draw=black,pattern={Lines[angle=-45,distance=2pt]},pattern color=spkred] (\p1) rectangle (\p2);
\path let \p1 = ($(7.7em,-6.2em)$),
          \p2 = ($(10.5em,-5.4em)$) in
    [draw=black,pattern={Lines[angle=-45,distance=2pt]},pattern color=spkred] (\p1) rectangle (\p2);
\node[font=\footnotesize,text=black,inner sep=0pt,anchor=base west] at ($(-6.3em,-5.9em)$) {+ Closing (\SI{0.5}{\second})};
\node[font=\tiny\selectfont,text=black,inner sep=0pt,anchor=base,align=left] at ($(0.6em,-5.25em)$) {0.3};
\node[font=\tiny\selectfont,text=black,inner sep=0pt,anchor=base,align=left] at ($(5.6em,-5.25em)$) {2.3};
\node[font=\tiny\selectfont,text=black,inner sep=0pt,anchor=base,align=left] at ($(7.7em,-5.25em)$) {3.3};
\node[font=\tiny\selectfont,text=black,inner sep=0pt,anchor=base,align=left] at ($(10.5em,-5.25em)$) {4.4};

\path let \p1 = ($(0.6em,-8.2em)$),
          \p2 = ($(10.5em,-7.4em)$) in
    [draw=black,pattern={Lines[angle=-45,distance=2pt]},pattern color=spkred] (\p1) rectangle (\p2);
\node[font=\footnotesize,text=black,inner sep=0pt,anchor=base west] at ($(-6.3em,-7.9em)$) {+ Closing (\SI{1.0}{\second})};
\node[font=\tiny\selectfont,text=black,inner sep=0pt,anchor=base,align=left] at ($(0.6em,-7.25em)$) {0.3};
\node[font=\tiny\selectfont,text=black,inner sep=0pt,anchor=base,align=left] at ($(10.5em,-7.25em)$) {4.4};

\node[font=\footnotesize\setlength{\baselineskip}{7pt},text=black,inner sep=0pt,anchor=base,align=center] at ($(13.5em,-2.1em)$) {\textbf{DER}\\{\scriptsize \textbf{(MI\,/\,FA)}}};
\node[font=\footnotesize\setlength{\baselineskip}{7pt},text=black,inner sep=0pt,anchor=base,align=center,yshift=-0.2em] at ($(13.5em,-4.1em)$) {35\,\%\\{\scriptsize (35\,\%\,/\,0\,\%)}};
\node[font=\footnotesize\setlength{\baselineskip}{7pt},text=black,inner sep=0pt,anchor=base,align=center,yshift=-0.2em] at ($(13.5em,-6.1em)$) {38\,\%\\{\scriptsize (30\,\%\,/\,8\,\%)}};
\node[font=\footnotesize\setlength{\baselineskip}{7pt},text=black,inner sep=0pt,anchor=base,align=center,yshift=-0.2em] at ($(13.5em,-8.1em)$) {13\,\%\\{\scriptsize (5\,\%\,/\,8\,\%)}};

\node[
  font=\footnotesize,
  rectangle,
  draw,
  rounded corners=2.5pt,
  align=center,
  inner sep=2.5pt,
] (bubble) at ($(18.7em,-4.2em)$) {\textbf{Proposed}\\\textbf{method}\\[4pt]Core: 10\,\%\\[-1pt]{\scriptsize (10\,\%\,/\,0\,\%)}\\[4pt]Pause: 25\,\%\\[-1pt]{\scriptsize (25\,\%\,/\,0\,\%)}};

\path[fill=white]
  ($(bubble.west)+(0.8pt,6pt)$)
  -- ($(bubble.west)+(-4pt,4pt)$)
  -- ($(bubble.west)+(0.8pt,2pt)$)
  -- cycle;

\draw[line join=miter]
  ($(bubble.west)+(0.5pt,6pt)$)
  -- ($(bubble.west)+(-4pt,4pt)$)
  -- ($(bubble.west)+(0.5pt,2pt)$);
\end{tikzpicture}%

%% file: figs/fill_hyp.tex
\begin{tikzpicture}[semithick,auto,
label/.style={
    draw=none,
    align=center,
    font=\footnotesize,
    inner sep=0,
    outer sep=0
},
]%
\def\tstart{0em}
\def\pstart{1em}
\def\pend{6.0em}
\def\tend{7.0em}
\definecolor{spkblue}{HTML}{0080B1}
\definecolor{spkred}{HTML}{E4002B}
\definecolor{spkgreen}{HTML}{06C755}

\path let \p1 = (\tstart,0em),
          \p2 = (\tend,-0.6em) in
    [pattern={Lines[angle=45,distance=2pt]},pattern color=spkblue] (\p1) rectangle (\p2);
\draw[black](\tstart,0em) -- (\tend,0em);
\draw[black](\tstart,-0.6em) -- (\tend,-0.6em);
\node[font=\footnotesize,text=black,inner sep=0pt,anchor=base west] at ($(-1.8em,-0.5em)$) {Ref};

\path let \p1 = ($(\tstart,-1.2em)$),
          \p2 = ($(\pstart,-1.8em)$) in
    [pattern={Lines[angle=-45,distance=2pt]},pattern color=spkred] (\p1) rectangle (\p2);
\path let \p1 = ($(\pend,-1.2em)$),
          \p2 = ($(\tend,-1.8em)$) in
    [pattern={Lines[angle=-45,distance=2pt]},pattern color=spkred] (\p1) rectangle (\p2);
\draw[black](\tend,-1.2em) -- (\pend,-1.2em)  -- (\pend,-1.8em) -- (\tend,-1.8em);
\draw[black](\tstart,-1.2em) -- (\pstart,-1.2em)  -- (\pstart,-1.8em) -- (\tstart,-1.8em);
\node[font=\footnotesize,text=black,inner sep=0pt,anchor=base west] at ($(-1.8em,-1.7em)$) {Hyp};

\draw[black,dash pattern=on 2pt off 1pt,line width=0.5pt,color=gray] (\pstart,-2.4em) -- (\pstart,0.6em);
\draw[black,dash pattern=on 2pt off 1pt,line width=0.5pt,color=gray] (\pend,-2.4em) -- (\pend,0.6em);
\draw[<->, >={Latex[length=3pt,width=3pt]}, black, line width=0.4pt]  (\pstart,-2.2em) -- (\pend,-2.2em);
\draw[decorate,decoration={brace, amplitude=5pt}] (\pstart,0.1em) -- (\pend,0.1em);
\node[font=\footnotesize,text=black,inner sep=0pt,anchor=base,xshift=-0.1em] at ($(0.5*\pstart+0.5*\pend,0.9em)$) {Active through $[a,b]$};
\node[font=\footnotesize,text=black,inner sep=0pt,anchor=base] at ($(0.5*\pstart+0.5*\pend,-3.1em)$) {$b-a~(\leq\tau)$};

\end{tikzpicture}%

%% file: figs/not_fill_hyp.tex
\begin{tikzpicture}[semithick,auto,
label/.style={
    draw=none,
    align=center,
    font=\footnotesize,
    inner sep=0,
    outer sep=0
},
]%
\def\tstart{0em}
\def\pstart{1.25em}
\def\pend{6.25em}
\def\tend{7.5em}
\definecolor{spkblue}{HTML}{0080B1}
\definecolor{spkred}{HTML}{E4002B}
\definecolor{spkgreen}{HTML}{06C755}

\def\ext{0.66em}
\path (-\ext,0) -- (\tend+\ext,0);

\path let \p1 = (\pstart,0em),
          \p2 = (\pend,-0.6em) in
    [draw,pattern={Lines[angle=45,distance=2pt]},pattern color=spkblue] (\p1) rectangle (\p2);

\path let \p1 = ($(\tstart,-1.2em)$),
          \p2 = ($(\pstart,-1.8em)$) in
    [pattern={Lines[angle=-45,distance=2pt]},pattern color=spkred] (\p1) rectangle (\p2);
\path let \p1 = ($(\pend,-1.2em)$),
          \p2 = ($(\tend,-1.8em)$) in
    [pattern={Lines[angle=-45,distance=2pt]},pattern color=spkred] (\p1) rectangle (\p2);
\draw[black](\tend,-1.2em) -- (\pend,-1.2em)  -- (\pend,-1.8em) -- (\tend,-1.8em);
\draw[black](\tstart,-1.2em) -- (\pstart,-1.2em)  -- (\pstart,-1.8em) -- (\tstart,-1.8em);

\draw[black,dash pattern=on 2pt off 1pt,line width=0.5pt,color=gray] (\pstart,-2.4em) -- (\pstart,0.6em);
\draw[black,dash pattern=on 2pt off 1pt,line width=0.5pt,color=gray] (\pend,-2.4em) -- (\pend,0.6em);
\draw[<->, >={Latex[length=3pt,width=3pt]}, black, line width=0.4pt]  (\pstart,-2.2em) -- (\pend,-2.2em);
\draw[decorate,decoration={brace, amplitude=5pt}] (\pstart,0.1em) -- (\pend,0.1em);
\node[font=\footnotesize,text=black,inner sep=0pt,anchor=base] at ($(0.5*\pstart+0.5*\pend,0.9em)$) {Active through $[a,b]$};
\node[font=\footnotesize,text=black,inner sep=0pt,anchor=base] at ($(0.5*\pstart+0.5*\pend,-3.1em)$) {$b-a~(\leq\tau)$};

\end{tikzpicture}%

%% file: figs/fill_ref.tex
\begin{tikzpicture}[semithick,auto,
label/.style={
    draw=none,
    align=center,
    font=\footnotesize,
    inner sep=0,
    outer sep=0
},
]%
\def\tstart{0em}
\def\pstart{1.3em}
\def\pend{6.3em}
\def\tend{7.6em}
\def\ext{0.6em}
\path (-\ext,0) -- (\tend+\ext,0);

\definecolor{spkblue}{HTML}{0080B1}
\definecolor{spkred}{HTML}{E4002B}
\definecolor{spkgreen}{HTML}{06C755}

\path let \p1 = ($(\pend,0em)$),
          \p2 = ($(\tend,-0.6em)$) in
    [pattern={Lines[angle=45,distance=2pt]},pattern color=spkblue] (\p1) rectangle (\p2);
\path let \p1 = ($(\tstart,0em)$),
          \p2 = ($(\pstart,-0.6em)$) in
    [pattern={Lines[angle=45,distance=2pt]},pattern color=spkblue] (\p1) rectangle (\p2);
\draw[black](\tend,0em) -- (\pend,0em)  -- (\pend,-0.6em) -- (\tend,-0.6em);
\draw[black](\tstart,0em) -- (\pstart,0em)  -- (\pstart,-0.6em) -- (\tstart,-0.6em);

\path let \p1 = ($(\tstart,-1.2em)$),
          \p2 = ($(\tend,-1.8em)$) in
    [pattern={Lines[angle=-45,distance=2pt]},pattern color=spkred] (\p1) rectangle (\p2);
\draw[black](\tstart,-1.2em) -- (\tend,-1.2em);
\draw[black](\tstart,-1.8em) -- (\tend,-1.8em);

\draw[black,dash pattern=on 2pt off 1pt,line width=0.5pt,color=gray] (\pstart,-2.4em) -- (\pstart,0.6em);
\draw[black,dash pattern=on 2pt off 1pt,line width=0.5pt,color=gray] (\pend,-2.4em) -- (\pend,0.6em);
\draw[<->, >={Latex[length=3pt,width=3pt]}, black, line width=0.4pt]
  (\pstart,0.4em) -- (\pend,0.4em);
\draw[decorate,decoration={brace, amplitude=5pt, mirror}] (\pstart,-1.9em) -- (\pend,-1.9em);
\node[font=\footnotesize,text=black,inner sep=0pt,anchor=base] at ($(0.5*\pstart+0.5*\pend,0.9em)$) {$b-a~(\leq\tau)$};
\node[font=\footnotesize,text=black,inner sep=0pt,anchor=base] at ($(0.5*\pstart+0.5*\pend,-3.1em)$) {Active through $[a,b]$};

\end{tikzpicture}%

%% file: figs/pause_attributable.tex
\begin{tikzpicture}[semithick,auto,
label/.style={
    draw=none,
    align=center,
    font=\footnotesize,
    inner sep=0,
    outer sep=0
},
]%

\def\xcap{-10.3em}
\def\xbase{1.85em}
\definecolor{spkblue}{HTML}{0080B1}
\definecolor{spkred}{HTML}{E4002B}
\definecolor{spkgreen}{HTML}{06C755}

\newcommand{\error}[5]{%
  \path let
      \p1 = (#1,#2+0.5pt),
      \p2 = (#3,#4-0.5pt)
  in
      node[
          minimum width={\x2-\x1},
          minimum height={\y1-\y2},
          inner sep=0pt,
          pattern={
              Hatch[
                angle=45,
                distance=2pt/sqrt(2),
                line width=0.4pt
              ]
          },
          pattern color=black!40,
          draw=none
      ] at ({(\x1+\x2)/2},{(\y1+\y2)/2}) {}

      node[
          font=\fontsize{6pt}{7pt}\selectfont,
          text=black,
          inner sep=0pt,
          align=center
      ] at ({(\x1+\x2)/2},{(\y1+\y2)/2})
      {\contour{white}{#5}}

      node[
          minimum width={\x2-\x1},
          minimum height={\y1-\y2},
          inner sep=0pt,
          draw=black,
          line width=0.5pt
      ] at ({(\x1+\x2)/2},{(\y1+\y2)/2}) {};
}
\def\ystart{0.6em}
\def\yend{-12.0em}
\draw[black!20,dash pattern=on 2pt off 1pt,line width=0.5pt] (0em,\yend) -- (0em,\ystart);
\draw[black!20,dash pattern=on 2pt off 1pt,line width=0.5pt] (\xbase,\yend) -- (\xbase,\ystart);
\draw[black!20,dash pattern=on 2pt off 1pt,line width=0.5pt] (2*\xbase,\yend) -- (2*\xbase,\ystart);
\draw[black!20,dash pattern=on 2pt off 1pt,line width=0.5pt] (3*\xbase,\yend) -- (3*\xbase,\ystart);
\draw[black!20,dash pattern=on 2pt off 1pt,line width=0.5pt] (4*\xbase,\yend) -- (4*\xbase,\ystart);
\draw[black!20,dash pattern=on 2pt off 1pt,line width=0.5pt] (5*\xbase,\yend) -- (5*\xbase,\ystart);
\draw[black!20,dash pattern=on 2pt off 1pt,line width=0.5pt] (6*\xbase,\yend) -- (6*\xbase,\ystart);
\draw[black!20,dash pattern=on 2pt off 1pt,line width=0.5pt] (7*\xbase,\yend) -- (7*\xbase,\ystart);
\draw[black!20,dash pattern=on 2pt off 1pt,line width=0.5pt] (8*\xbase,\yend) -- (8*\xbase,\ystart);
\draw[black!20,dash pattern=on 2pt off 1pt,line width=0.5pt] (9*\xbase,\yend) -- (9*\xbase,\ystart);

\node[font=\footnotesize,text=black,inner sep=0pt,anchor=base west] at (\xcap,-0.85em) {Reference};
\node[font=\scriptsize,text=black,inner sep=0pt,anchor=base east] at ($(-0.3em,-0.5em)$) {Speaker 1};
\node[font=\scriptsize,text=black,inner sep=0pt,anchor=base east] at ($(-0.3em,-1.2em)$) {Speaker 2};
\path let \p1 = ($(0em,0em)$),
          \p2 = ($(5*\xbase,-0.6em)$) in
    [pattern={Lines[angle=45,distance=2pt]},draw,pattern color=spkblue] (\p1) rectangle (\p2);
\path let \p1 = ($(5*\xbase,0em)$),
          \p2 = ($(7*\xbase,-0.6em)$) in
    [pattern={Dots[distance=2.0pt, radius=0.4pt]},pattern color=spkblue] (\p1) rectangle (\p2);
\path let \p1 = ($(6*\xbase,-0.7em)$),
          \p2 = ($(7*\xbase,-1.3em)$) in
    [draw=black,pattern={Lines[angle=45,distance=2pt]},pattern color=spkblue] (\p1) rectangle (\p2);
\path let \p1 = ($(7*\xbase,0em)$),
          \p2 = ($(8*\xbase,-0.6em)$) in
    [draw,pattern={Lines[angle=45,distance=2pt]},pattern color=spkblue] (\p1) rectangle (\p2);
\draw[decorate,decoration={brace, amplitude=5pt}] (5*\xbase,0.1em) -- (7*\xbase,0.1em);
\node[font=\scriptsize,text=black,inner sep=0pt,anchor=base] at (6*\xbase,1.0em) {Reference pause ($\leq\tau$)};

\node[font=\footnotesize,text=black,inner sep=0pt,anchor=base west] at (\xcap,-3.35em) {Hypothesis};
\node[font=\scriptsize,text=black,inner sep=0pt,anchor=base east] at ($(-0.3em,-3.0em)$) {Speaker 1};
\node[font=\scriptsize,text=black,inner sep=0pt,anchor=base east] at ($(-0.3em,-3.7em)$) {Speaker 2};
\path let \p1 = ($(1*\xbase,-2.5em)$),
          \p2 = ($(2*\xbase,-3.1em)$) in
    [pattern={Lines[angle=-45,distance=2pt]},draw,pattern color=spkred] (\p1) rectangle (\p2);
\path let \p1 = ($(2*\xbase,-3.2em)$),
          \p2 = ($(3*\xbase,-3.8em)$) in
    [draw=black,pattern={Lines[angle=-45,distance=2pt]},pattern color=spkred] (\p1) rectangle (\p2);
\path let \p1 = ($(2*\xbase,-2.5em)$),
          \p2 = ($(4*\xbase,-3.1em)$) in
    [pattern={Dots[distance=2.0pt, radius=0.4pt]},pattern color=spkred] (\p1) rectangle (\p2);
\path let \p1 = ($(4*\xbase,-2.5em)$),
          \p2 = ($(9*\xbase,-3.1em)$) in
    [draw,pattern={Lines[angle=-45,distance=2pt]},pattern color=spkred] (\p1) rectangle (\p2);
\draw[decorate,decoration={brace, amplitude=5pt}] (2*\xbase,-2.4em) -- (4*\xbase,-2.4em);
\node[font=\scriptsize,text=black,inner sep=0pt,anchor=base] at (3*\xbase,-1.5em) {Hypothesis pause ($\leq\tau$)};

\draw[black!40](\xcap,-4.25em) -- (9*\xbase,-4.25em);
\draw[black!40](\xcap,-7.55em) -- (9*\xbase,-7.55em);
\draw[black!40](\xcap,-9.95em) -- (9*\xbase,-9.95em);

\node[font=\footnotesize,text=black,inner sep=0pt,anchor=base west,align=left] at (\xcap,-6.7em) {Instantaneous\\error counts};
\node[font=\scriptsize,text=black,inner sep=0pt,anchor=base east] at ($(-0.2em,-5.2em)$) {$M(t)$};
\node[font=\scriptsize,text=black,inner sep=0pt,anchor=base east] at ($(-0.2em,-6.1em)$) {$F(t)$};
\node[font=\scriptsize,text=black,inner sep=0pt,anchor=base east] at ($(-0.2em,-7.0em)$) {$C(t)$};
\error{0*\xbase}{-4.7em}{1*\xbase}{-5.3em}{$1$}
\error{2*\xbase}{-6.5em}{3*\xbase}{-7.1em}{$1$}
\error{3*\xbase}{-4.7em}{4*\xbase}{-5.3em}{$1$}
\error{5*\xbase}{-5.6em}{6*\xbase}{-6.2em}{$1$}
\error{6*\xbase}{-6.5em}{7*\xbase}{-7.1em}{$1$}
\error{8*\xbase}{-5.6em}{9*\xbase}{-6.2em}{$1$}

\node[font=\footnotesize,text=black,inner sep=0pt,anchor=base west,align=left] at (\xcap,-9.5em) {Pause-attributable\\candidate};
\node[font=\scriptsize,text=black,inner sep=0pt,anchor=base east] at ($(-0.2em,-8.5em)$) {$P_M(t;\tau)$};
\node[font=\scriptsize,text=black,inner sep=0pt,anchor=base east] at ($(-0.2em,-9.4em)$) {$P_F(t;\tau)$};
\error{2*\xbase}{-8em}{4*\xbase}{-8.6em}{$1$}
\error{5*\xbase}{-8.8em}{7*\xbase}{-9.4em}{$1$}

\node[font=\footnotesize,text=black,inner sep=0pt,anchor=base west,align=left] at (\xcap,-11.9em) {Pause-\\attributable};
\node[font=\scriptsize,text=black,inner sep=0pt,anchor=base east] at ($(-0.2em,-10.9em)$) {$M_\mathrm{pause}(t;\tau)$};
\node[font=\scriptsize,text=black,inner sep=0pt,anchor=base east] at ($(-0.2em,-11.8em)$) {$F_\mathrm{pause}(t;\tau)$};
\error{3*\xbase}{-10.4em}{4*\xbase}{-11.0em}{$1$}
\error{5*\xbase}{-11.3em}{6*\xbase}{-11.9em}{$1$}

\end{tikzpicture}%

%% file: figs/synthetic/complete.tex
\begin{tikzpicture}[semithick,auto,
label/.style={
    draw=none,
    align=center,
    font=\footnotesize,
    inner sep=0,
    outer sep=0
},
]%

\definecolor{spkblue}{HTML}{0080B1}
\definecolor{spkred}{HTML}{E4002B}
\definecolor{spkgreen}{HTML}{06C755}

\path let \p1 = ($(3.35em,0em)$),
          \p2 = ($(4.1em,-1em)$) in
    [pattern={Lines[angle=45,distance=2pt]},pattern color=spkblue] (\p1) rectangle (\p2);
\path let \p1 = ($(0em,0em)$),
          \p2 = ($(0.75em,-1em)$) in
    [pattern={Lines[angle=45,distance=2pt]},pattern color=spkblue] (\p1) rectangle (\p2);
\draw[black](4.1em,0em) -- (3.35em,0em)  -- (3.35em,-1em) -- (4.1em,-1em);
\draw[black](0em,0em) -- (0.75em,0em)  -- (0.75em,-1em) -- (0em,-1em);
\node[font=\footnotesize,text=black,inner sep=0pt,anchor=base west] at ($(-1.8em,-0.7em)$) {Ref};

\path let \p1 = ($(0em,-1.5em)$),
          \p2 = ($(4.1em,-2.5em)$) in
    [pattern={Lines[angle=-45,distance=2pt]},pattern color=spkred] (\p1) rectangle (\p2);
\draw[black](0em,-1.5em) -- (4.1em,-1.5em);
\draw[black](0em,-2.5em) -- (4.1em,-2.5em);
\node[font=\footnotesize,text=black,inner sep=0pt,anchor=base west] at ($(-1.8em,-2.3em)$) {\contour{white}{Hyp}};

\draw[black,dash pattern=on 2pt off 1pt,line width=0.5pt,color=gray] (0.75em,1em) -- (0.75em,0em);
\draw[black,dash pattern=on 2pt off 1pt,line width=0.5pt,color=gray] (3.35em,1em) -- (3.35em,0em);
\draw[<->, >={Latex[length=3pt,width=3pt]}, black, line width=0.4pt]
  (0.75em,0.5em) -- node[midway,anchor=base,yshift=5pt,font=\footnotesize, inner sep=0pt] {$d \scriptstyle{\,(\leq d_{\max})}$}
  (3.35em,0.5em);

\end{tikzpicture}%

%% file: figs/synthetic/partial.tex
\begin{tikzpicture}[semithick,auto,
label/.style={
    draw=none,
    align=center,
    font=\footnotesize,
    inner sep=0,
    outer sep=0
},
]%

\definecolor{spkblue}{HTML}{0080B1}
\definecolor{spkred}{HTML}{E4002B}
\definecolor{spkgreen}{HTML}{06C755}

\path let \p1 = ($(4.05em,0em)$),
          \p2 = ($(4.8em,-1em)$) in
    [pattern={Lines[angle=45,distance=2pt]},pattern color=spkblue] (\p1) rectangle (\p2);
\path let \p1 = ($(0em,0em)$),
          \p2 = ($(0.75em,-1em)$) in
    [pattern={Lines[angle=45,distance=2pt]},pattern color=spkblue] (\p1) rectangle (\p2);
\draw[black](4.8em,0em) -- (4.05em,0em)  -- (4.05em,-1em) -- (4.8em,-1em);
\draw[black](0em,0em) -- (0.75em,0em)  -- (0.75em,-1em) -- (0em,-1em);

\path let \p1 = ($(0em,-1.5em)$),
          \p2 = ($(0.75em,-2.5em)$) in
    [pattern={Lines[angle=-45,distance=2pt]},pattern color=spkred] (\p1) rectangle (\p2);
\path let \p1 = ($(1.85em,-1.5em)$),
          \p2 = ($(2.95em,-2.5em)$) in
    [draw=black,pattern={Lines[angle=-45,distance=2pt]},pattern color=spkred] (\p1) rectangle (\p2);
\path let \p1 = ($(4.05em,-1.5em)$),
          \p2 = ($(4.8em,-2.5em)$) in
    [pattern={Lines[angle=-45,distance=2pt]},pattern color=spkred] (\p1) rectangle (\p2);
    
\draw[black](4.8em,-1.5em) -- (4.05em,-1.5em)  -- (4.05em,-2.5em) -- (4.8em,-2.5em);
\draw[black](0em,-1.5em) -- (0.75em,-1.5em)  -- (0.75em,-2.5em) -- (0em,-2.5em);

\draw[<->, >={Latex[length=3pt,width=3pt]}, black, line width=0.4pt]
  (0.75em,0.5em) -- node[midway,anchor=base,yshift=5pt,font=\footnotesize, inner sep=0pt] {$d \scriptstyle{\,(\leq d_\mathrm{max})}$}
  (4.05em,0.5em);
  
\draw[black,dash pattern=on 2pt off 1pt,line width=0.5pt,color=gray] (0.75em,1em) -- (0.75em,0em);
\draw[black,dash pattern=on 2pt off 1pt,line width=0.5pt,color=gray] (4.05em,1em) -- (4.05em,0em);
\draw[<->, >={Latex[length=3pt,width=3pt]}, black, line width=0.4pt]
  (0.75em,-2.1em) -- node[midway,anchor=base,yshift=5pt,font=\footnotesize, inner sep=2pt] {$\frac{d}{3}$}
  (1.85em,-2.1em);
\draw[<->, >={Latex[length=3pt,width=3pt]}, black, line width=0.4pt]
  (2.95em,-2.1em) -- node[midway,anchor=base,yshift=5pt,font=\footnotesize, inner sep=2pt] {$\frac{d}{3}$}
  (4.05em,-2.1em);
  
\end{tikzpicture}%

%% file: figs/synthetic/shift.tex
\begin{tikzpicture}[semithick,auto,
label/.style={
    draw=none,
    align=center,
    font=\footnotesize,
    inner sep=0,
    outer sep=0
},
]%

\definecolor{spkblue}{HTML}{0080B1}
\definecolor{spkred}{HTML}{E4002B}
\definecolor{spkgreen}{HTML}{06C755}

\path let \p1 = ($(-0.2em,0em)$),
          \p2 = ($(1.5em,-1em)$) in
    [pattern={Lines[angle=45,distance=2pt]},pattern color=spkblue] (\p1) rectangle (\p2);
\path let \p1 = ($(5.6em,0em)$),
          \p2 = ($(7.7em,-1em)$) in
    [pattern={Lines[angle=45,distance=2pt]},pattern color=spkblue] (\p1) rectangle (\p2);
\draw[black](-0.2em,0em) -- (1.5em,0em)  -- (1.5em,-1em) -- (-0.2em,-1em);
\draw[black](7.7em,0em) -- (5.6em,0em)  -- (5.6em,-1em) -- (7.7em,-1em);

\path let \p1 = ($(6.6em,-1.5em)$),
          \p2 = ($(7.7em,-2.5em)$) in
    [pattern={Lines[angle=-45,distance=2pt]},pattern color=spkred] (\p1) rectangle (\p2);
\path let \p1 = ($(-0.2em,-1.5em)$),
          \p2 = ($(2.5em,-2.5em)$) in
    [pattern={Lines[angle=-45,distance=2pt]},pattern color=spkred] (\p1) rectangle (\p2);
\draw[black](-0.2em,-1.5em) -- (2.5em,-1.5em)  -- (2.5em,-2.5em) -- (-0.2em,-2.5em);
\draw[black](7.7em,-1.5em) -- (6.6em,-1.5em)  -- (6.6em,-2.5em) -- (7.7em,-2.5em);

\draw[->, >={Latex[length=3pt,width=3pt]},black,line width=0.4pt]
  (-0.2em,0.5em) -- node[midway,anchor=base,yshift=5pt,xshift=-1pt,font=\footnotesize,inner sep=0pt] {$u_\text{left}$}
  (1.5em,0.5em);
\draw[<->, >={Latex[length=3pt,width=3pt]},black, line width=0.4pt]
  (1.5em,0.5em) -- node[midway,anchor=base,yshift=5pt,font=\footnotesize,inner sep=0pt] {$d \scriptstyle{\,(\leq d_\mathrm{max})}$}
  (5.6em,0.5em);
\draw[<-, >={Latex[length=3pt,width=3pt]},black, line width=0.4pt]
  (5.6em,0.5em) -- node[midway,anchor=base,yshift=5pt,xshift=1pt,font=\footnotesize,inner sep=0pt] {$u_\text{right}$}
  (7.7em,0.5em);

\draw[black,dash pattern=on 2pt off 1pt,line width=0.5pt,color=gray] (1.5em,1em) -- (1.5em,0em);
\draw[black,dash pattern=on 2pt off 1pt,line width=0.5pt,color=gray] (5.6em,1em) -- (5.6em,0em);
\draw[black,dash pattern=on 2pt off 1pt,line width=0.5pt,color=gray] (2.5em,-0.5em) -- (2.5em,-1.5em);
\draw[black,dash pattern=on 2pt off 1pt,line width=0.5pt,color=gray] (5.6em,-2.3em) -- (5.6em,-1em);

\draw[<->, >={Latex[length=3pt,width=3pt]}, black, line width=0.4pt]
  (1.5em,-0.8em) -- node[midway,above,font=\footnotesize, inner sep=2pt] {$\delta$}
  (2.5em,-0.8em);
\draw[<->, >={Latex[length=3pt,width=3pt]}, black, line width=0.4pt]
  (5.6em,-2.1em) -- node[midway,above,font=\footnotesize, inner sep=2pt] {$\delta$}
  (6.6em,-2.1em);

\end{tikzpicture}%

%% file: figs/synthetic/long.tex
\begin{tikzpicture}[semithick,auto,
label/.style={
    draw=none,
    align=center,
    font=\footnotesize,
    inner sep=0,
    outer sep=0
},
]%

\definecolor{spkblue}{HTML}{0080B1}
\definecolor{spkred}{HTML}{E4002B}
\definecolor{spkgreen}{HTML}{06C755}

\path let \p1 = ($(5.75em,0em)$),
          \p2 = ($(6.5em,-1em)$) in
    [pattern={Lines[angle=45,distance=2pt]},pattern color=spkblue] (\p1) rectangle (\p2);
\path let \p1 = ($(0em,0em)$),
          \p2 = ($(0.75em,-1em)$) in
    [pattern={Lines[angle=45,distance=2pt]},pattern color=spkblue] (\p1) rectangle (\p2);
\draw[black](6.5em,0em) -- (5.75em,0em)  -- (5.75em,-1em) -- (6.5em,-1em);
\draw[black](0em,0em) -- (0.75em,0em)  -- (0.75em,-1em) -- (0em,-1em);

\path let \p1 = ($(0em,-1.5em)$),
          \p2 = ($(6.5em,-2.5em)$) in
    [pattern={Lines[angle=-45,distance=2pt]},pattern color=spkred] (\p1) rectangle (\p2);
\draw[black](0em,-1.5em) -- (6.5em,-1.5em);
\draw[black](0em,-2.5em) -- (6.5em,-2.5em);

\draw[black,dash pattern=on 2pt off 1pt,line width=0.5pt,color=gray] (0.75em,1em) -- (0.75em,0em);
\draw[black,dash pattern=on 2pt off 1pt,line width=0.5pt,color=gray] (5.75em,1em) -- (5.75em,0em);
\draw[<->, >={Latex[length=3pt,width=3pt]}, black, line width=0.4pt]
  (0.75em,0.5em) -- node[midway,anchor=base,yshift=5pt,font=\footnotesize, inner sep=0pt] {$d \scriptstyle{\,(\in(d_{\max},2d_{\max}])}$}
  (5.75em,0.5em);

\end{tikzpicture}%